\documentclass[11pt, a4paper]{article}
\usepackage{amsmath}%
\usepackage{amsfonts}%
\usepackage{amssymb}%
\usepackage{enumerate}
\usepackage[font=small]{caption}
\usepackage{hyperref}
\usepackage{picture}
\usepackage{color}
\newtheorem{theorem}{Theorem} [section]

\newtheorem{lemma}[theorem]{Lemma}

\newtheorem{observation}[theorem]{Observation}

\numberwithin{equation}{section}
\newenvironment{proof}[1][Proof]{\textbf{#1.} }{\ \rule{0.5em}{0.5em}}%
\usepackage{a4wide} 
\newcommand{\op}[1]{\text{\normalfont{\sc #1}}}
\begin{document}
\title{Buffer Overflow Management with Class Segregation}
\author{Kamal Al-Bawani \footnote{Department of Computer Science, RWTH Aachen University, \texttt{kbawani@cs.rwth-aachen.de}. Supported in part by NRW State within the B-IT Research School.} 
        \and Alexander Souza \footnote{Department of Computer Science, Humboldt University of Berlin, \texttt{souza@informatik.hu-berlin.de}}}
\date{}
\maketitle
\begin{abstract}
\noindent We consider a new model for buffer management of network switches with Quality of Service (QoS) requirements. A stream of packets, each attributed with a value representing its Class of Service (CoS), arrives over time at a network switch and demands a further transmission. The switch is equipped with multiple queues of limited capacities, where each queue stores packets of one value only. The objective is to maximize the total value of the transmitted packets (i.e., the weighted throughput).

We analyze a natural greedy algorithm, \op{greedy}, which sends in each time step a packet with the greatest value. For general packet values $(v_1 < \cdots < v_m)$, we show that \op{greedy} is $(1+r)$-competitive, where $r = \max_{1\le i \le m-1} \{v_i/v_{i+1}\}$. Furthermore, we show a lower bound of $2 - v_m / \sum_{i=1}^m v_i$ on the competitiveness of any deterministic online algorithm. In the special case of two packet values (1 and $\alpha > 1$), \op{greedy} is shown to be optimal with a competitive ratio of $(\alpha + 2)/(\alpha + 1)$.
\end{abstract}
\section{Introduction}
In most of today's networks, the concept of Quality of Service (QoS) is implemented by assigning different numerical values to packets in order to provide them with differentiated levels of service (priorities). A packet value corresponds to its Class of Service (CoS). QoS guarantees are essential for network streams that, for example, are latency sensitive or require fixed bit rates (e.g., voice-over-IP and IP-television streams).

We study the following \emph{online} problem of packet buffering in QoS networks. A sequence of packets arrives over time, each with a certain value, at the input ports of a network switch. Upon its arrival, a packet $p$ is either accepted or rejected. Accepted packets are stored in queues of restricted capacities placed at the input ports. At each time step, at most one packet is extracted from a queue and sent out of the switch. If preemption is allowed, an enqueued packet may be dropped before it is sent. Rejected and preempted packets are deleted, and the total value of the transmitted (sent) packets defines the benefit of the sequence. We aim at maximizing this benefit.

In probabilistic analysis of network traffic, packet arrivals are often assumed to be Poisson processes. This assumption has been widely undermined due to burstiness of network traffic (i.e., packets arrive in bursts rather than in smooth Poisson-like flows) \cite{Paxson95}. Therefore, it became common to consider buffering problems with \emph{competitive analysis} \cite{Borodin} which assumes no probability distribution on the network traffic. Specifically, the performance of an online algorithm \op{alg} is measured by comparing its benefit to the benefit of an optimal offline algorithm \op{opt}, i.e., an algorithm that knows the entire input sequence at the outset and determines a solution maximizing the benefit. \op{alg} is called \emph{$c$-competitive} if, for any input sequence, the inequality $\op{opt}/\op{alg} \leq c$ holds, where \op{opt} and \op{alg} denote the respective total benefits.
\paragraph{New Model: Motivation and Notation.} In this paper, we bring into focus a new model of $n$ queues and $m$ packet values in which each queue is assigned exactly one value. Let $V = \{ v_1 < \dots < v_m \}$ be the set of $m$ non-negative \emph{packet values}, and $Q = \left\{q_1, \ldots, q_n \right\}$ be the set of $n$ queues. A packet that is of value $v_{i}$ is denoted a \emph{$v_{i}$-packet}. Each queue is assigned a value. If a queue $q_k$ is assigned a value $v_i$, all packets stored in $q_k$ are $v_i$-packets, and $q_k$ is called a \emph{$v_{i}$-queue}. Each queue $q_k$ has its own capacity $B_k$, i.e., $q_k$ can store up to $B_k$ packets.

This model is motivated by two applications:
\begin{enumerate}[(1)]
\item  In the state-of-the-art design of switching devices, e.g. of Cisco switches~\cite{cisco124}, each queue is devoted to store packets of one CoS-value only. In this case, $n=m$ and the $i$-th queue stores packets of the $i$-th value, for $1 \leq i \leq m$. 
\item In OFDMA-based wireless networks~\cite{Li03}, a base station that serves $n$ users is devised with $n$ queues, such that the $i$-th queue stores packets destined to the $i$-th user, for $1 \leq i \leq n$. To reflect the QoS requirements in this model, one may assign each user/queue (and hence each packet belonging to this user) a value that corresponds to its CoS. In this case, it can be that $n \neq m$ as multiple users may be of the same value or a value is assigned to none of the users.
\end{enumerate}
Time is discretized into steps of unit length. An input sequence consists of arrive and send events, where an arrive event corresponds to the arrival of a new packet and a send event corresponds to the transmission of an enqueued packet. Packets arrive over time, and each arriving packet is attributed with a unique arrive time, a value $v_i \in V$ and a destination queue $q_k \in Q$, such that $q_k$ is a $v_i$-queue. We assume that these attributes are determined by an adversary. Arrive events occur at non-integral times (within a time step) and an arriving packet is either \emph{admitted} to the queue or \emph{rejected}. Send events occur at integral times (at the end of each time step), and in each send event, exactly one packet from a non-empty queue, if any, is \emph{transmitted} (\emph{sent}). Preemption does not make sense in this model as all packets inserted into the same queue have the same value. Thus, a packet that is accepted must eventually be sent. The \emph{benefit} of an event sequence is defined as the sum of values of the sent packets, and our goal is to maximize this benefit. 

Any online and offline algorithm in this model can be considered to be \textit{diligent}: It must send a packet if it has a non-empty queue at sending time. Moreover, since accepting a packet of one value does not interfere with packets of other values, we may extend the concept of diligence and assume that any algorithm must accept a packet arriving at a queue if this queue has residual capacity. Notice that each algorithm, which is not diligent can be transformed into a diligent algorithm, without changing the benefit. This is because the diligent algorithm can send everything the other algorithm sends, but earlier; and it can accept everything the other algorithm accepts, but also earlier.
\paragraph{Related Work.} The problem of packet buffering has been considerably studied in the recent years under a wide variety of models. Despite its relevance in practice, our model with class segregation has not been considered theoretically before. We concentrate here only on previous multi-queue models. For a comprehensive and up-to-date survey of competitive buffer management, we refer the reader to \cite{Goldwasser10}. \newline
\indent Our model generalizes the unit-valued model which consists of $m$ queues of the same capacity $B$ and where all packets have unit value. Azar and Richter \cite{Azar03} showed that any diligent policy is $2$-competitive, and that no deterministic policy can be better than $(2 - 1/m)$-competitive for $B = 1$ or $(1.366 - \Theta(1/m))$-competitive for arbitrary $B$. Albers and Schmidt \cite{Albers04} improved these results and gave a policy called \op{semi-greedy} that is $1.889$-competitive for any $B$ with $m \gg B$, and $1.857$-competitive for $B = 2$. They also improved the lower bound to $e/(e-1) \approx 1.582$ for any $B$ if $m \gg B$. Azar and Litichevskey \cite{Azar04b} achieved this lower bound for large $B$, giving an algorithm with a competitive ratio that tends to $e/(e-1)$. If randomization is allowed, Azar and Richter \cite{Azar03} gave a $e/(e-1) \approx 1.582$-competitive policy for $B > \log m$, and a lower bound of $(1.46 - \Theta(1/m))$ for $B = 1$. Albers and Schmidt \cite{Albers04} 
improved the lower bound to $1.466$ for any $B$ and large $m$, and $1.231$ for $m = 2$. Recently, Bienkowski and Madry \cite{Bienkowski08} gave a policy for $m=2$ and any $B$ with the optimal competitive ratio of $1.231$.

On the other hand, our model is a special case of the general-valued multi-queue model in which each of $m$ \op{FIFO} queues can store up to $B$ packets of different values. Azar and Richter \cite{Azar03} started by giving a preemptive algorithm with competitive ratio of $(4 + 2 \ln(\alpha))$ for general packet values in $[1, \alpha]$, and a $2.6$-competitive algorithm for the special case of two packet values, $1$ and $\alpha>1$. After that, they provided in \cite{Azar04} a $3$-competitive policy called \op{transmit-largest-head} (\op{tlh}) for the general-valued case. Recently, Itoh and Takahashi \cite{Itoh06} improved the analysis of \op{tlh} and proved $(3 - 1/\alpha)$-competitiveness.
\paragraph{Our Contribution.} We study a natural greedy algorithm, \op{greedy}, which sends in each time step a packet with the greatest value. We analyze \op{greedy} under two model variants: the general-valued variant with $m$ packet values $\{ v_1, \dots, v_m \}$, and the two-valued variant with two packet values, 1 and $\alpha > 1$. The latter variant is especially interesting since in practical applications there are often only two types of packet streams: best-effort packets (given value $1$) and priority packets (given value $\alpha$). 

In the general-valued case, we show that \op{greedy} is $(1+r)$-competitive, where $r = \max_{1\le i \le m-1} \{v_i/v_{i+1}\}$. Thus, the competitive ratio of \op{greedy} is strictly below 2, and it even tends to 1 as $r$ tends to 0. For example, choosing the packet values to be powers of 2 makes \op{greedy} 3/2-competitive. 

Furthermore, we show a lower bound of $2 - v_m / \sum_{i=1}^m v_i$ on the competitiveness of any deterministic online algorithm. In fact, one can easily show that \op{greedy} is asymptotically optimal, even with a competitive ratio near 2. Notice that when $r$ tends to 1, the differences between packet values will shrink to 0 and thus the lower bound will tend to $2-1/m$, which is asymptotically 2.

Finally, in the two-valued case, \op{greedy} is shown to be optimal with a competitive ratio of $(\alpha + 2)/(\alpha + 1)$.
\section{An Online Algorithm \op{greedy}}
We define \op{greedy} as follows. At arrive events, \op{greedy} accepts packets of any value until the respective queue becomes full. At send events, it serves the non-empty queue, if any, with the highest packet value, i.e., a $v_{i}$-packet is sent only when all $v_{j}$-queues are empty, for $i < j \le m$. If more than one queue of the highest value are non-empty, \op{greedy} arbitrarily chooses one of them for sending.

We start by showing that \op{greedy} is 2-competitive in the general-valued case, and $(\alpha +1)/\alpha$-competitive in the two-valued case.
\begin{theorem}
\label{theorem_generalcaseratio} 
The competitive ratio of \op{greedy} is at most $2$ for general packet values, and $(\alpha +1)/\alpha$ for two packet values, 1 and $\alpha > 1$.
\end{theorem}
\begin{proof}
 The idea of the proof is to modify \op{opt} in a manner that does not decrease its benefit and which ensures that the queues of both \op{opt} and \op{greedy} remain identical at any time. Thus, to show a competitive ratio of $c$, it suffices to show that the benefit of the modified \op{opt} is in any send event at most $c$ times the benefit of \op{greedy}. (This idea is used in \cite{Li05} in the context of packet buffering with deadlines.)
 
 Fix a time step $t$ in which both \op{opt} and \op{greedy} have identical queues. Let $\op{greedy}(t)$ (resp., $\op{opt}(t)$) denote the value of the packet(s) sent by \op{greedy} (resp., $\op{opt}$) in $t$. Since \op{opt} is assumed to be diligent as well, it will accept in $t$ as many packets as \op{greedy} accepts of each value. Thus, the queues of the two algorithms remain identical at the end of $t$. Now, consider the send event of $t$: (1) If \op{opt} and \op{greedy} send a packet from the same queue, the queues of the two algorithms remain identical and $\op{opt}(t) = \op{greedy}(t)$. (2) If \op{greedy} sends a packet from a $v_i$-queue and \op{opt} sends a packet from a $v_j$-queue, it must hold that $v_j \le v_i$, since both algorithms have the same queue contents and \op{greedy} gives priority to packets of the greatest value. In this case, we modify \op{opt} by allowing it to additionally send a packet of the same $v_i$-queue that \op{greedy} sends from. Thus, $\op{opt}(t) = v_j + v_i \le 2 v_i = 2 \cdot \op{greedy}(t)$. Moreover, we insert a packet in \op{opt}'s $v_j$-queue, so the queues of \op{opt} and \op{greedy} become identical again. Clearly, these modifications on \op{opt} can only increase its benefit.
 
 With a similar argument, one can show a competitive ratio of $(\alpha +1)/\alpha$ for the case of two packet values, 1 and $\alpha > 1$.
\end{proof}
\\ \\
Next, we show improved competitive ratios for \op{greedy} when our model is slightly restricted, so that only one queue is devoted for each packet value, and all queues have one capacity $B$. Due to these restrictions, the competitive ratio of \op{greedy} is shown to be strictly below 2 in the general-valued case, and $(\alpha +2)/(\alpha+1)$ in the two-valued case.

\subsection{The General-valued Case}
Let $r = \max_{1\le i \le m-1} \{v_i/v_{i+1}\}$. The following theorem shows that \op{greedy} is $(1+r)$-competitive.
\begin{theorem}
\label{theorem_generalcaseratio} 
The competitive ratio of \op{greedy} is at most $1+r$.
\end{theorem}
We fix an event sequence $\sigma$. Let \op{greedy}'s benefit on $\sigma$ be denoted as $\op{greedy}(\sigma)$, and the benefit of \op{opt} as $\op{opt}(\sigma)$. Furthermore, let $A_i$ and $A^*_i$ denote the total number of $v_{i}$-packets accepted by \op{greedy} and \op{opt}, respectively. Hence, $\op{greedy}(\sigma) = \sum_{i=1}^{m} v_{i} A_i$, and $\op{opt}(\sigma) = \sum_{i=1}^{m} v_{i} A^*_i$. 

We begin by showing that \op{opt} and \op{greedy} send the same number of $v_{m}$-packets.
\begin{lemma}
\label{lemma_ag} 
$A^*_m = A_m$.
\end{lemma}%
\begin{proof}
By definition of \op{greedy}, $v_{m}$-packets enjoy absolute priority at sending. Hence, the number of $v_{m}$-packets that \op{greedy} sends is maximum, i.e., $A_m \geq A^*_m$.\newline
\indent Assume that $A_m$ becomes greater than $A^*_m$ for the first time at arrive event $t$. This means that \op{opt} rejects at $t$ a $v_{m}$-packet that \op{greedy} accepts. Hence, \op{opt}'s $v_{m}$-queue was full immediately before $t$ but \op{greedy}'s had at least one vacancy. Let $j=1, \cdots, m-1$. Since $A_m = A^*_m$ immediately before $t$, there must have been a send event $t'$ before $t$ where \op{opt} sent a $v_{j}$-packet while its $v_{m}$-queue was not empty. Change \op{opt}'s schedule by sending a $v_{m}$-packet at $t'$ instead of the $v_{j}$-packet. Clearly, this yields an increase in \op{opt}'s benefit and the rejected $v_{m}$-packet at time $t$ can be accepted.
\end{proof}
\\\ \\\
The following lemma shows an upper bound on the total number of packets that \op{opt} accepts but \op{greedy} rejects.
\begin{lemma}
\label{centrallemma}
For any $1\leq i \leq m-1$, 
\begin{equation}
\nonumber \sum_{j=i}^{m-1} \left(A^*_j - A_j\right) \leq \sum_{j=i+1}^{m} A_j.
\end{equation}
\end{lemma}
\begin{proof}
For any $1 \leq i \leq m-1$, we define the following notion of \textit{time interval}. A time interval $I$ ends with a send event, and the next time interval starts with the first arrive event after the end of $I$. We call a time interval $I$ \textit{red interval} (or $r$-interval) if the value of any packet sent by \op{greedy} in $I$ is in $\left\{v_{i},\ldots, v_{m}\right\} \subseteq V$, and \textit{green interval} (or $g$-interval) if the value of any packet sent by \op{greedy} in $I$ is in $\left\{v_{1},\ldots, v_{i-1}\right\} \subseteq V$ or $I$ contains send events in which \op{greedy} does not send any packet. We partition $\sigma$ into $r$- and $g$-intervals such that no two consecutive intervals are of the same color. Clearly, this partitioning is feasible.\newline
\indent For the rest of the proof, let $v_j$-packet denote any packet of value $v_{i},\ldots, v_{m}$. The following observation follows from the definition of $g$-intervals; otherwise, \op{greedy} would send $v_{j}$-packets in a $g$-interval and thus it is no longer a $g$-interval.
\begin{observation}
\label{remark_0interval} 
In any $g$-interval, any $v_{j}$-queue of \op{greedy} is empty and no $v_{j}$-packets arrive.
\end{observation}
Let $A_j(I)$ (resp., $A^*_j(I)$) denote the total number of $v_j$-packets accepted by \op{greedy} (resp., \op{opt}) in time interval $I$, and let $\Re$ denote the set of all $r$-intervals. Given Observation \ref{remark_0interval}, it follows that $A_j = \sum_{I \in \Re} A_j(I)$ and $A^*_j =  \sum_{I \in \Re} A^*_j(I)$. Hence, it suffices to prove the lemma for any $r$-interval. Thus, we fix an $r$-interval $I$ and show that $\sum_{j=i}^{m-1} ( A^*_j(I) - A_j(I) ) \leq \sum_{j=i+1}^{m} A_j(I)$.

Let $\delta_{j}(t)$ denote the total number of $v_{j}$-packets sent by \op{greedy} between the first and the $t$-th event of $I$, inclusive. Let $b_{j}(t)$ denote the size of \op{greedy}'s $v_{j}$-queue right after the $t$-th event of $I$. By Observation \ref{remark_0interval}, interval $I$ starts with \op{greedy}'s $v_{j}$-queue empty. Thus, if no further $v_{j}$-packets arrive in $I$ after event $t$, we get that $A_j(I) = \delta_{j}(t) + b_{j}(t)$. For \op{opt}, we define $\delta_{j}^{*}(t)$ as the total number of $v_{j}$-packets \emph{that arrive in $I$} and are sent by \op{opt} between the first and the $t$-th event of $I$, inclusive; and $b_{j}^{*}(t)$ as the number of packets \emph{that arrive in $I$} and still reside in \op{opt}'s $v_{j}$-queue right after the $t$-th event of $I$. Thus, if no further $v_{j}$-packets arrive in $I$ after event $t$, we also get that $A_j^{*}(I) = \delta_{j}^{*}(t) + b_{j}^{*}(t)$.

We now define a potential function $\varphi:\mathbb{Z} \mapsto \mathbb{Z}$ as follows.
\begin{equation}
\nonumber \varphi(t) = \sum_{j=i}^{m} \left( \delta_{j}(t)+ b_{j}(t) \right) + \sum_{j=i+1}^{m-1} \left( \delta_{j}(t)+ b_{j}(t) \right) - \sum_{j=i}^{m-1} \left(\delta^{*}_{j}(t) + b^{*}_{j}(t)\right)
\end{equation}
Clearly, it suffices to prove that $\varphi(t) \geq 0$ for any $t \geq 0$. We conduct an induction proof over the number of arrive and send events of $I$. (Notice that $\varphi(t)$ changes on arrive and send events only.) For the induction basis, $\varphi(0) = 0$ as nothing arrives in $I$ before its first event. Assume that $\varphi(t) \geq 0$ for $t \le k-1$. We will show that $\varphi(k) \geq 0$. 

First, assume that the $k$-th event is a send event. Since $I$ is an $r$-interval, \op{greedy} sends a $v_{l}$-packet, where $i \leq l \leq m$. We show that $\delta_{j}(k)+ b_{j}(k) = \delta_{j}(k-1)+ b_{j}(k-1)$, for any $j = i, \cdots, m$. This is obvious for $j \neq l$ as nothing has changed in $v_{j}$-queue since event $k-1$. For $j = l$, $\delta_{j}(k) = \delta_{j}(k-1)+1$ as the number of sent $v_{j}$-packets increases by 1, and $b_{j}(k)= b_{j}(k-1)-1$ as the size of $v_{j}$-queue decreases by 1. Thus, $\delta_{j}(k)+ b_{j}(k) = \delta_{j}(k-1)+ b_{j}(k-1)$, as required. Similarly, if \op{opt} sends a $v_{j}$-packet that arrives in $I$, one can show that $\delta^{*}_{j}(k)+ b^{*}_{j}(k) = \delta^{*}_{j}(k-1)+ b^{*}_{j}(k-1)$, for any $j = i, \cdots, m-1$. Therefore, $\varphi(k) = \varphi(k-1) \geq 0$.

Assume now that the $k$-th event is an arrive event. We first observe that in arrive events, $\delta_{j}(k) = \delta_{j}(k-1)$ and $\delta^{*}_{j}(k) = \delta^{*}_{j}(k-1)$ as nothing has been sent since the last event. Moreover, since $I$ is an $r$-interval, the arriving packet must be a $v_{l}$-packet, where $i \leq l \leq m$. Now, we distinguish between four cases: (1) Both \op{opt} and \op{greedy} reject the packet, (2) both accept it, (3) \op{greedy} accepts and \op{opt} rejects, or (4) \op{opt} accepts and \op{greedy} rejects. The potential function does clearly not change in the first case. In the second and third case, $v_{l}$-queue increases by one in \op{greedy} and by at most one in \op{opt} (it does not increase in \op{opt} in the third case). Thus, $\varphi(k)$ will increase by at least 1 due to change in \op{greedy}'s queue (it increases by 2 if $i+1 \le l \le m-1$) and will decrease by at most 1 due to change in \op{opt}'s queue. Hence, $\varphi(k) \ge \varphi(k-1)$.

Now, consider the last case where \op{opt} accepts but \op{greedy} rejects. First, we rewrite $\varphi(k)$ as
\begin{align}
 \nonumber \varphi(k) = \sum_{j=i}^{m} \delta_{j}(k) + \sum_{j=i}^{m} b_{j}(k) + \sum_{j=i+1}^{m-1} \left( \delta_{j}(k)+ b_{j}(k) \right) - \sum_{j=i}^{m-1} \delta^{*}_{j}(k) - \sum_{j=i+1}^{m-1} b^{*}_{j}(k) - b^{*}_{i}(k).
\end{align}
Since the arriving $v_{l}$-packet is rejected by \op{greedy}, \op{greedy}'s $v_{l}$-queue must be full. Thus, $\sum_{j=i}^{m} b_{j}(k) \geq B$. However, the size of \op{opt}'s $v_{i}$-queue can never exceed $B$. Hence 
\begin{equation}
 \label{ineq1} \sum_{j=i}^{m} \left(b_{j}(k)\right) - b^{*}_{i}(k) \geq 0.
\end{equation}
Let $\eta$ be the number of send events in interval $I$ up to the current arrive event $k$. Clearly, $\sum_{j=i}^{m-1} \delta^{*}_{j}(k) \leq \eta$. By definition of $r$-intervals, \op{greedy} has sent exactly $\eta$ packets of value $v_{j} \in \left\{v_{i},\ldots, v_{m}\right\}$. Thus, $\sum_{j=i}^{m} \delta_{j}(k)=\eta$. Hence,
\begin{equation}
 \label{ineq2} \sum_{j=i}^{m} \delta_{j}(k) - \sum_{j=i}^{m-1} \delta^{*}_{j}(k) \geq 0.
\end{equation}
Recall that $b^{*}_{j}$ counts \op{opt}'s $v_{j}$-packets that it accepts only in $I$. Thus, for $i+1 \le j \le m-1$, at least $b^{*}_{j}(k)$ packets must arrive in $I$ up to event $k$. Since \op{greedy} starts interval $I$ with an empty $v_{j}$-queue, and since $b^{*}_{j}(k) \leq B$, \op{greedy} must be able to accept at least $b^{*}_{j}(k)$ packets of all $v_{j}$-packets arriving in $I$. Thus, recalling that $\delta_{j}(k)+ b_{j}(k)$ represents the number of $v_{j}$-packets that \op{greedy} accepts in $I$ up to $k$, we get that $\delta_{j}(k)+ b_{j}(k) \geq b^{*}_{j}(k)$, and hence,
\begin{equation}
 \label{ineq3} \sum_{j=i+1}^{m-1} \left( \delta_{j}(k)+ b_{j}(k) \right) - \sum_{j=i+1}^{m-1} b^{*}_{j}(k)  \geq 0.
\end{equation}
By summing up inequalities \ref{ineq1} - \ref{ineq3}, we get that $\varphi(k) \geq 0$, and thus the lemma follows.
\end{proof}
\\\ \\\
Next, we show a weighted version of Lemma \ref{centrallemma} when $i=1$, which is, along with Lemma \ref{lemma_cruciallemma}, essential to improve the competitive ratio of \op{greedy} from 2 to $1+r$.
\begin{lemma}
\label{lemma_secondmainlemma} 
We have that 
\begin{equation}
\nonumber \sum_{j=1}^{m-1} v_j (A^*_j - A_j) \leq \sum_{j=1}^{m-1} v_{j} A_{j+1}.
\end{equation}
\end{lemma}
\begin{proof}
For simplicity of exposition, let $m=4$. By Lemma \ref{centrallemma},
\begin{align*}
 i&=1: & (A^*_1 - A_1) + (A^*_2 - A_2) + (A^*_3 - A_3) & & &\le & A_2 + A_3 + A_4 &, \\
 i&=2: & (A^*_2 - A_2) + (A^*_3 - A_3) & & &\le & A_3 + A_4 &, \\
 i&=3: & (A^*_3 - A_3) & & &\le & A_4 &.
\end{align*}
Recall that $v_i < v_{i+1}$, for all $1 \le i \le m-1$. Thus, multiplying both sides of ($i = 1$) by $v_1$, both sides of ($i = 2$) by ($v_2 - v_1$), both sides of ($i = 3$) by ($v_3 - v_2$), and adding up all the resulting inequalities, we get
\begin{equation*}
 v_1 (A^*_1 - A_1) + v_2 (A^*_2 - A_2) + v_3 (A^*_3 - A_3) \le v_1 A_2 + v_2 A_3 + v_3 A_4.
\end{equation*}
Clearly, the argument above can be generalized for any $m \ge 2$.
\end{proof}
\begin{lemma}
\label{lemma_cruciallemma} 
We have that 
\begin{equation}
\nonumber \frac{\sum_{j=1}^{m-1} v_j A_{j+1}}{\sum_{j=1}^{m-1} v_{j+1} A_{j+1}} \le r.
\end{equation}
\end{lemma}
\begin{proof} 
Let $r = v_k/v_{k+1}$, for some $1 \le k \le m-1$. Hence, by definition of $r$, for any $1 \le j \le m-1$, 
\begin{equation}
 \nonumber \frac{v_j}{v_{j+1}} \le \frac{v_k}{v_{k+1}},
\end{equation}
which can be rewritten as $v_{k+1} \cdot v_j \le v_k \cdot v_{j+1}$. Thus, multiplying by $A_{j+1}$ and then summing over all $j$, we get
\begin{equation}
 \nonumber v_{k+1} \sum_{j=1}^{m-1} v_j A_{j+1} \le v_k \sum_{j=1}^{m-1} v_{j+1} A_{j+1},
\end{equation}
from which the lemma follows.
\end{proof}
\\ \\
\begin{proof}[Proof of Theorem~\ref{theorem_generalcaseratio}] The competitive ratio of \op{greedy} is concluded as follows.  
\begin{align} 
 \nonumber \frac{\op{opt}(\sigma)}{\op{greedy}(\sigma)} 
 = 1 + \frac{\sum_{j=1}^{m-1} v_j (A^*_j - A_j)}{\sum_{j=1}^{m} v_j A_j}
 \leq 1 + \frac{\sum_{j=1}^{m-1} v_j A_{j+1}}{\sum_{j=1}^{m-1} v_{j+1} A_{j+1}} \leq 1 + r,
\end{align}
where the equality follows from Lemma \ref{lemma_ag}, while the first and second inequalities follow from Lemma \ref{lemma_secondmainlemma} and Lemma \ref{lemma_cruciallemma}, respectively.
\end{proof}
\subsection{The Two-Valued Case}
The following theorem shows an upper bound of $(\alpha +2)/(\alpha+1)$ on the competitive ratio of \op{greedy}.
\begin{theorem}
The competitive ratio of \op{greedy} is at most $(\alpha +2)/(\alpha+1)$.
\label{thm:two-values}
\end{theorem}
We assume that no packets arrive after the first time in which all queues of \op{greedy} become empty. This assumption is without loss of generality as we can partition the event sequence $\sigma$ into phases such that each phase satisfies this assumption, and queues of \op{greedy} and \op{opt} are all empty at the start and end of the phase. Then, it is sufficient to show the theorem on any arbitrary phase. Consider for example the creation of the first phase. Let $t$ be the first time in which all queues of \op{greedy} become empty. We postpone the packets arriving after $t$ until \op{opt}'s queues are empty as well, say at time $t'$. Thus, \op{opt} and \op{greedy} are both empty at time $t'$. This change can increase the benefit of \op{opt} only. Clearly, $t'$ defines the end of the first phase, and the next arrive time after $t'$ defines the start of the second phase. The remaining of $\sigma$ can be further partitioned in the same way.

Let $b(t)$ and $b^*(t)$ denote the total size of the queues of \op{greedy} and \op{opt}, respectively, at time $t$. (Later on, we will slightly abuse these notations to denote the queue sizes right after an event $t$.) The following lemma shows that at any time, the contents of \op{opt} is at most $B$ larger than the contents of \op{greedy}.
\begin{lemma}
\label{lem_queuesizes} 
At any time $t$, $b^{*}(t) \leq b(t) + B$.
\end{lemma}
\begin{proof}
We do an induction proof over the number of send and arrive events to show that $\varphi(t) = b(t) - b^{*}(t) + B \geq 0$, for any $t \geq 0$. Obviously, all queues are empty at time 0 and thus the claim holds. We now assume it holds for all events up to the $i$-th. If the $i$-th event is a send event, the only critical case is when \op{greedy} sends a packet while \op{opt} does not. In this case, \op{opt}'s queues must be empty as \op{opt} is assumed to be diligent, and hence the claim holds.

Assume now that the $i$-th event is an arrive event. Here, the only critical case is when \op{opt} accepts the arriving packet while \op{greedy} rejects it. In this case, we know that at least one queue of \op{greedy} is full. Thus, $b(i) \geq B$. Since the total size of \op{opt}'s queues cannot exceed $2B$, $b(i)-b^{*}(i)+B \geq 0$.
\end{proof}
\\ \\
Besides Lemma \ref{centrallemma}, the following lemma gives another upper bound on the number of 1-packets that \op{opt} accepts but \op{greedy} rejects.
\begin{lemma}
\label{lem_ineq1} 
$ A^{*}_{1} - A_{1} \leq A_{1}$.
\end{lemma}
\begin{proof}
Intuitively, the lemma holds if $A^{*}_{1} \leq A_{1}$. Thus, assume that $A^{*}_1 > A_1$. Therefore, there must exist at least one 1-packet $p$ that is accepted by \op{opt} but rejected by \op{greedy}. At the arrival of $p$, \op{greedy}'s 1-queue must be full, and thus $A_{1} \geq B$.

Let \op{greedy} send its last packet at the $i$-th send event. Clearly, $b(i) = 0$. Since $\sigma$ is a one-phase sequence, \op{greedy}'s queues become empty at $i$ for the first time. Thus, the total number of packets accepted by \op{greedy} is exactly $i$, i.e., $A_{1}+A_{\alpha} = i$. Moreover, \op{opt} sends at most $i$ packets up to event $i$. Hence, $A^{*}_{1}+A^{*}_{\alpha} \leq i + b^{*}(i) = A_{1}+A_{\alpha} + b^{*}(i)$. By Lemma \ref{lemma_ag} and \ref{lem_queuesizes}, $A^{*}_{\alpha} = A_{\alpha}$ and $b^{*}(i) \leq B$. Therefore, $A^{*}_{1} \leq A_{1} + B$. Finally, given that $A_{1} \geq B$, $A^{*}_{1} \leq A_{1} + A_{1}$.
\end{proof}
\\ \\
\begin{proof}[Proof of Theorem~\ref{thm:two-values}] Given Lemma \ref{lemma_ag}, if $A^{*}_{1} \leq A_{1}$, then \op{greedy} is obviously 1-competitive. Thus, assume that $A^{*}_{1} > A_{1}$. Therefore, the competitive ratio of \op{greedy} is concluded as follows.  
\begin{align} 
 \nonumber \frac{\op{opt}(\sigma)}{\op{greedy}(\sigma)} 
  & = 1 + \frac{A^{*}_{1} - A_{1}}{\alpha A_{\alpha}+A_{1}} \leq 1 + \frac{A^*_1 - A_1}{\alpha(A^*_1 - A_1) + (A^*_1 - A_1)} = \frac{\alpha + 2}{\alpha+1},
\end{align}
where the first equality follows from Lemma \ref{lemma_ag}, and the inequality from Lemma \ref{centrallemma} and Lemma \ref{lem_ineq1}.
\end{proof} 
\section{A Lower Bound}
Finally, we show a lower bound on the competitive ratio of any deterministic online algorithm.
\begin{theorem}
The competitive ratio of any deterministic online algorithm is at least $2- v_{m}/(\sum_{i=1}^{m} v_{i})$.
\end{theorem}
\begin{proof}
Let \op{alg} be any deterministic online algorithm, which will have to compete against an offline algorithm \op{adv}. As argued earlier, we may assume that both \op{alg} and \op{adv} are diligent. Furthermore, assume that there exists exactly one queue of each value, and that all queues are of unit size, i.e., $B_k=1$, for $1 \leq k \leq m$.

We construct an adversarial instance $\sigma$ in the following way: In each time step $1 \leq i \leq m$, packets of \emph{distinct} values arrive. We denote the set of values of the packets arriving in step $i$ by $V_i \subseteq V$. Let $V_1 = V$. By diligence, we may assume that \op{alg} accepts all packets in step 1. Let $s_1$ be the packet-value sent by \op{alg} in the send event of step $1$. Since \op{alg} is online and deterministic, this value is well-defined. Define $V_2 = V_1 - \{ s_1 \}$ and let $s_2$ be the packet-value sent by \op{alg} if $V_2$ arrives in step $2$. Now define $V_3 = V_2 - \{ s_2 \}$ and in general $V_{t+1} = V_{t} - \{ s_t \}$ for $t = 1, \dots, m-1$. Clearly, $V_1 \supset V_2 \supset \dots \supset V_m$.

Observe that all packets arriving in the time steps $2, \dots, m$ are of values corresponding to non-empty queues in \op{alg}. So \op{alg} must reject all of these packets. Hence $\op{alg}(\sigma) = \sum_{i = 1}^m v_i$ since \op{alg} accepts all packets arriving in step $1$.

Now we define the algorithm \op{adv}: It accepts all packets $V_1$ in step $1$, but sends the packet with value $s_2$ in the send event of that step. Since $s_2$ is still in the queues of \op{alg} in step $2$, we have $s_2 \in V_2$. Now \op{adv} accepts the packet with value $s_2$ in step $2$ and rejects all other packets. In the send event of step $2$, \op{adv} sends a packet with value $s_3$. In general, in the send event of step $t$, it sends a packet with value $s_{t+1}$ for $t = 1, \dots, m-1$; and, for $t = 2, \dots, m$, \op{adv} accepts the packet with value $s_t$ in step $t$ and rejects all other packets.

Observe that before the send event of step $m$, \op{adv} still has one packet of each value in its queues and no further packets will arrive. So, \op{adv} sends those packets in the send events of steps $m, \dots, 2m-1$, in any order. Therefore, if \op{alg} sends one packet of value $v$, \op{adv} sends two packets of value $v$, except for the one packet with value $s_1$. Thus we have $\op{adv}(\sigma) = 2 \cdot \sum_{i = 1}^m v_i - s_1 \ge 2 \cdot \sum_{i = 1}^m v_i - v_m$.

Therefore,
\[
\frac{\op{opt}(\sigma)}{\op{alg}(\sigma)} \geq  \frac{\op{adv}(\sigma)}{\op{alg}(\sigma)} \geq \frac{2 \cdot \sum_{i=1}^{m} v_{i} \;\;- v_{m}}{\sum_{i=1}^{m} v_{i}} = 2- \frac{v_{m}}{\sum_{i=1}^{m} v_{i}}.
\]
\end{proof}
\section*{Acknowledgements}
 The authors thank Toshiya Itoh and anonymous referees for their helpful remarks.

\end{document}